\documentclass[runningheads]{llncs}

\usepackage{booktabs}   
\usepackage{subcaption} 

\usepackage{graphpap,amscd,mathrsfs,graphicx,lscape,dsfont,bm,url,color}
\usepackage{amsmath, amsfonts, amssymb}
\usepackage{verbatim}
\usepackage{parcolumns}
\usepackage{mathtools, cuted, IEEEtrantools}

\usepackage{hyperref}
\usepackage[capitalise,nameinlink]{cleveref}
\usepackage{bm}
\usepackage{bold-extra}
\usepackage{threeparttable}
\usepackage{tabularx}
\usepackage{adjustbox}
\usepackage{xspace}
\usepackage[inline]{enumitem}
\usepackage{balance}

\newcommand{\qqpi}[2]{[\![#2]\!]_{#1}}

\newcommand{\projectname}{\emph{OptTyper}\xspace}

\newcommand{\restate}[1]{\textsc{Restatement of #1}. \hspace*{1pt} \it}

\definecolor{Maroon}{cmyk}{0.4,0.87,0.68,0.32} 
\definecolor{RoyalBlue}{rgb}{0.0,0.14,0.6}
\definecolor{mygreen}{rgb}{0.45,0.62,0.51}
\definecolor{UscGold}{rgb}{1.0,0.8,0.0}
\definecolor{mygray}{rgb}{0.35,0.35,0.35}
\definecolor{mypurple}{rgb}{0.69,0.50,0.63}
\definecolor{myrose}{rgb}{0.58,0.50,0.63}
\usepackage{color-edits}
\addauthor{ivp}{Maroon}
\addauthor{adg}{RoyalBlue}
\addauthor{cas}{mygreen}
\addauthor{ebr}{UscGold}

\usepackage{listings}
\lstdefinelanguage{JavaScript}{
keywords={const, typeof, new, true, false, catch, function, 
  return, null, catch, switch, var, if, in, while, do, else, 
  case, break, class, export,throw, implements, import, this,
  exports, interface, readonly, from, public},
keywordstyle=\color{myrose},
ndkeywords={boolean, string, number, any, Array, Window, Event, Array<any>, Object},
ndkeywordstyle=\color{Maroon},
identifierstyle=\color{black},
sensitive=false,
escapechar=!,
comment=[l]{//},
morecomment=[s]{/*}{*/},
commentstyle=\color{mygray}\ttfamily,
morestring=[b]',
morestring=[b]"
}
 
\newlength{\listingindent}  
\usepackage{mathrsfs}
\lstdefinelanguage{Fake}{
    keywords={function, return},
    keywordstyle=\bfseries,
    sensitive=false,
    escapechar=!,
    numbers=none
}

\graphicspath{ {figs/} }

\begin{document}

\title{Probabilistic Type Inference by Optimising Logical and Natural Constraints\thanks{This work was supported by Microsoft Research through
its PhD Scholarship Programme.}}

\author{Irene Vlassi Pandi\inst{1} \and
Earl T. Barr\inst{2} \and
Andrew D. Gordon\inst{1,3}\and
Charles Sutton\inst{1,4,5}}
\authorrunning{I. Vlassi Pandi et al.}
%
\institute{University of Edinburgh, Edinburgh, UK \\ \email{irene.vp@ed.ac.uk} \and 
University College London, London, UK \\ \email{e.barr@ucl.ac.uk}\and Microsoft Research Cambridge, Cambridge, UK \\ \email{adg@microsoft.ac.uk}
\and
The Alan Turing Institute, London, UK \and Google AI, Mountain View, USA \\\email{charlessutton@google.com}}

\maketitle
\begin{abstract}\label{sec:abstract}
    We present a new approach to the type inference problem for dynamic languages. Our goal is to combine \emph{logical} constraints, that is, deterministic information from a type system, with \emph{natural} constraints, that is, uncertain statistical information about types learnt from sources like identifier names. To this end, we introduce a framework for probabilistic type inference that combines logic and learning: logical constraints on the types are extracted from the program, and deep learning is applied to predict types from surface-level code properties that are statistically associated. The foremost insight of our method is to constrain the predictions from the learning procedure to respect the logical constraints, which we achieve by relaxing the logical inference problem of type prediction into a continuous optimisation problem.
	%
    We build a tool called \projectname to predict missing types for TypeScript files.
	\projectname 
	combines a continuous interpretation of logical constraints derived by classical static analysis of TypeScript code, with natural constraints obtained from a deep learning model, which learns naming conventions for types from a large codebase. 
	By evaluating \projectname, we show that the combination of logical and natural constraints yields a large improvement in performance over either kind of information individually and achieves a 4\% improvement over the state-of-the-art.
	%
\end{abstract}

\section{Introduction}

Statically-typed programming languages aim to enforce correctness and safety properties
on programs by guaranteeing constraints on program behaviour.
A large scale user-study
suggests that programmers
benefit from type safety~\cite{hanenberg14}; use of types has also been
shown to prevent field bugs~\cite{gao17}.
However, type safety comes at a cost: these languages often require explicit type annotations,
which imposes the burden of declaring and maintaining these annotations on the programmer.
Strongly statically-typed, usually functional languages, like Haskell or ML,
offer type inference procedures that reduce
the cost of explicitly writing types but come with
a steep learning curve~\cite{tirronen15}.

Dynamically typed languages, which either lack or do not require type
annotations, are relatively more popular~\cite{meyerovich12}.  Initially
designed for quick scripting or rapid prototyping, these languages
have begun reaching the limits of what can be achieved without the help of type
annotations, as witnessed by the heavy industrial investment in and proliferation of static type systems for these languages (TypeScript~\cite{typescript} and Flow~\cite{flow} are just two
examples).
Retrofit for dynamic languages, these type systems include gradual~\cite{siek06}
and optional type systems~\cite{bracha2004pluggable}.  Like classical type systems, these type systems
require annotations to provide benefits.  Hence, reducing the annotation type tax
for dynamic languages remains an open research topic.

\subsection{Probabilistic Type Inference}

Probabilistic type inference
has recently been proposed as an attempt to reduce the burden
of writing and maintaining type
annotations~\cite{wei20,raychev15,hellendoorn18}.
Just as the availability of large data sets has transformed artificial intelligence,
the increased volume of publicly available source code, through
code repositories like GitHub~\cite{github}
or GitLab~\cite{gitlab},
enables a new class of applications that leverage statistical
patterns in large codebases~\cite{allamanis17}.
Specifically, for type inference, machine learning
allows us to develop less strict type inference systems
that learn to predict types from uncertain information,
such as comments, names, and lexical context,
even when traditional type inference procedures
fail to infer a useful type.

The classic literature on conventional type systems takes great care to demonstrate
that type inference only suggests sound types~\cite{DBLP:journals/jcss/Milner78,Pierce2002}.
Probabilistic type inference is not in conflict with classical type inference but complements it.
There are settings, like TypeScript, where correct type inference is too imprecise.
In these settings, probabilistic type inference helps the human in the loop to move a
partially typed codebase---one lacking so many type annotations that classical type inference
can make little progress---to a sufficiently annotated state that classical type inference can take over and finish the job.

Two examples of probabilistic type inference tools are JSNice~\cite{raychev15},
which uses probabilistic graphical models to statistically infer types of identifiers
in programs written in JavaScript,
and DeepTyper~\cite{hellendoorn18}, which targets TypeScript via deep learning techniques.
However, none explicitly models the underlying
type inference rules, and, thus, their predictions ignore useful type constraints.
Most recently, Wei et al.~\cite{wei20} introduced LambdaNet to exploit both type constraints and name usage information by using graph neural networks~\cite{allamanis17}. LambdaNet, however, does not constrain 
\emph{the output} of the network to satisfy these constraints;
this must be learnt automatically from data,
and there is no guarantee that the resulting model will respect the type constraints at test time. 
Indeed, in practice, we observe that LambdaNet
produces annotations that do not respect the learnt 
logical relationships (see \cref{fig:LambdaNet}).

Recognising this problem TypeWriter~\cite{pradel19}  combines a deep learning model and a type checker for probabilistic type prediction for Python.
TypeWriter proposes enumerating the top-ranked predictions from a neural type predictor, and subsequently filtering those that do not type-check.
A key difference is that \projectname incorporates both sorts of constraint into a single prediction step, whereas TypeWriter validates the probabilistically predicted types in a second, potentially expensive, step using a type checker.


\subsection{Our Contribution}

In our view, probabilistic type inference should be considered as a constrained problem, as it makes no sense to suggest types that violate type constraints. To respect this principle, we propose \projectname (from ``optimising for optional types''), a principled framework for probabilistic type inference that couples hard, \textit{logical} type constraints with soft constraints drawn from structural, \textit{natural} patterns into a single optimisation problem. While, in theory, there is the option of filtering out the incorrect predictions, our framework goes beyond that; our composite optimisation serves as a communication channel between the two different sources of information. 

Current type inference systems rely
on one of two sources of information
\begin{enumerate}[label=(\Roman*)]
	\item \emph{Logical} Constraints on type annotations that follow from the type system.
	      These are the constraints used by standard deterministic approaches for static type inference.
	\item \emph{Natural} Constraints are statistical constraints on type annotations
	      that can be inferred from relationships between types and surface-level properties such as names and lexical context.
	      These constraints can be learned by applying machine learning to large codebases.
\end{enumerate}
Our goal is to improve the accuracy of probabilistic type
inference by combining both kinds of constraints into a single analysis, unifying logic and learning into a single framework.
We start with a formula that defines the logical constraints on the types of a set of identifiers in the program,
and a machine learning model, such as a deep neural network, that probabilistically predicts the type of each identifier.

The key idea behind our methods is a
\emph{continuous relaxation} of
the logical constraints~\cite{hajek98}.
This means that we relax the logical formula into a continuous function by relaxing type environments
to probability matrices and defining
a continuous semantic interpretation of logical expressions.
The relaxation has a special property, namely,
that when this continuous function is maximised with respect
to the relaxed type environment, we obtain a discrete type environment
that satisfies the original constraints.
The benefit of this relaxation is that logical constraints
can now be combined with the probabilistic
predictions of type assignments that are
produced by machine learning methods.

More specifically, this allows us to define a continuous function over the continuous version of the type environment
that sums the logical and natural constraints.
And once we have a continuous function, we can optimise it:
we set up an optimisation problem that returns the most natural type assignment for a
program, while at the same time respecting type constraints produced by traditional type inference.
\begin{figure}[!t]
    \hspace{-1cm}%
    \def\svgwidth{1.10\linewidth}
    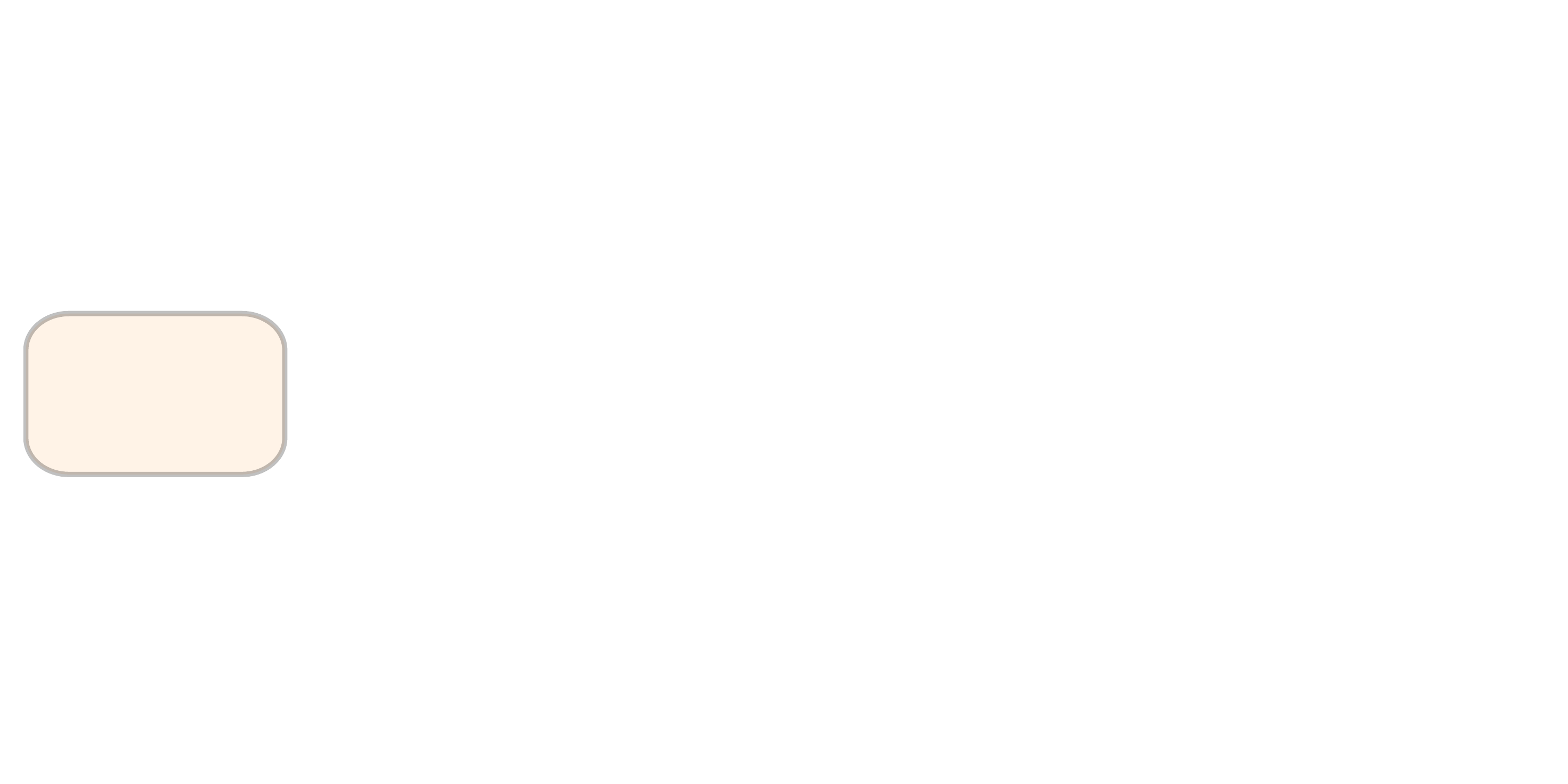
    \caption{Overview of general framework that combines logical
    and natural constraints in a single optimisation problem.}\label{fig:overview}
\end{figure}
Our main contributions follow:
\begin{itemize}[label=\raisebox{0.25ex}{\tiny$\bullet$}]
	\item We introduce a general, principled framework that uses soft logic to combine logical and natural constraints for type inference,
	      based on transforming a type inference procedure into a single numerical optimisation problem.
	\item We instantiate this framework in \projectname, a probabilistic type inference tool for TypeScript.
	\item We evaluate \projectname and find that combining logical and natural constraints has better performance than either alone. Further, \projectname outperforms state-of-the-art systems,
	      LambdaNet, DeepTyper and JSNice.
	\item We show how \projectname achieves its high performance by combining \textit{Logical} and \textit{Natural} constraints as an optimisation problem at test time.
\end{itemize}


\section{General Framework for Probabilistic Type Inference}\label{sec:framework}

\begin{figure}[!t]
    \hspace{-0.8cm}%
	\def\svgwidth{1.4\linewidth}
	 \scalebox{.80}{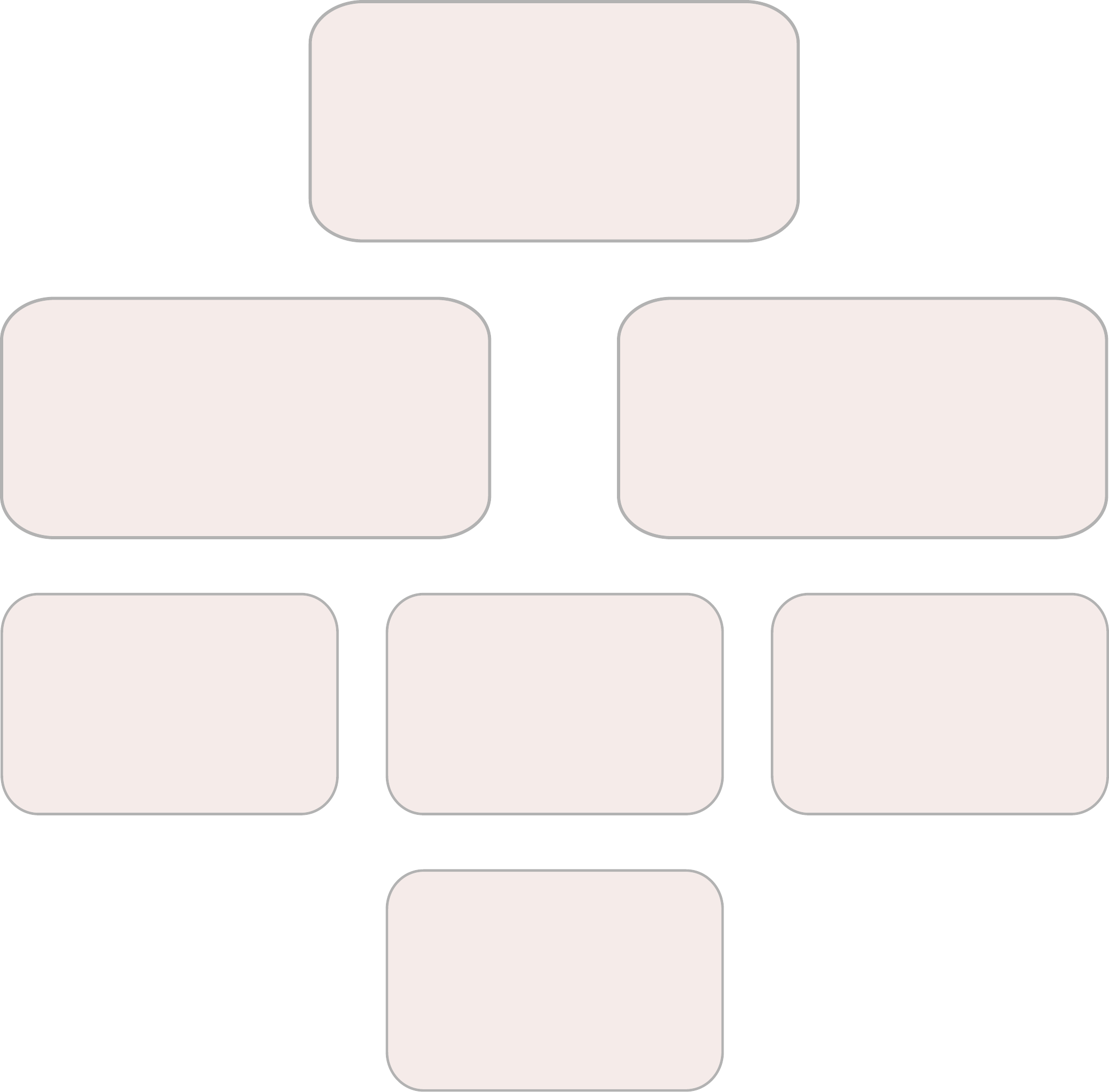}
	\caption{An overview of the three type inference procedures via a minimal example.}
\label{fig:fullexample}
\end{figure}
This section introduces our general framework, shown in~\cref{fig:overview}
which we instantiate in the next section by building a tool for
predicting types in TypeScript.
\cref{fig:fullexample} illustrates our general framework through a simplified running example of predicting types.

 Our input is a minimal  function
		with no type annotations on its parameters or result.
		%
		%
		Our goal is to exploit both logical and natural constraints to suggest
		more specific types.
		To begin, in Box (a), we propose fresh type annotations  \textcolor{mygreen}{\texttt{START}} and \textcolor{mygreen}{\texttt{END}} (uppercasing the identifier) for each parameter and \textcolor{mygreen}{\texttt{ADDNUM}} for the return type.
		We insert these annotations into the function's definition.
		Our \emph{logical constraints} on these types represent knowledge obtained
		by a symbolic analysis of the code in the function's body.
		In our example, the use of a binary operation implies that the two parameter types are equal.
		Box (c) shows a minimal set of logical constraints that state
		that \textcolor{RoyalBlue}{\texttt{addNum}}'s two operands have the same type.
		In general, the logical constraints can be much more complex than our simple example.
		If we only have logical constraints, we cannot tell
		whether \lstinline{string} or \lstinline{number} is a better solution,
		and so may fall back to the type \lstinline{any}.
		The crux of our approach is to take into account \emph{natural constraints};
		that is, statistical properties learnt from a source code corpus that seek to
		capture human intention.
		In particular, we use a machine learning model to capture naming conventions over types.
		We represent the solution space for our logical or natural constraints
		or their combination
		as a $V \times T$ matrix $P$ of the form in Box (b):
		each row vector is a discrete probability distribution
		over our universe of $T=3$ concrete types
		(\lstinline{number}, \lstinline{string}, and \lstinline{any})
		for one of our $V=3$ identifiers.
		Box (d) shows the natural constraints $\mathcal{M}$ induced by the identifier names
		for the parameters and the function name itself.
		Intuitively, Box (d) shows that a programmer
		is more likely to name a variable \lstinline{start} or \lstinline{end}
		if she intends to use it as a \lstinline{number} than as a \lstinline{string}.
		%
		%
		We can relax the boolean constraint
		to a numerical function on probabilities as shown in Box (c).
		When we numerically optimise the resulting expression,
		we obtain the matrix in Box (e);
		it predicts that both variables are strings with high probability.
		Although the objective function is symmetric
		between \lstinline{string} and \lstinline{number},
		the solution in (e) is asymmetric because it depends on the initialisation
		of the optimiser.
		Finally, Box (f) shows an optimisation objective that
		combines both sources of information:
		$E$ consists of the logical constraints
		and each probability vector $\mu_v$ (the row of $\mathcal{M}$ for $v$)
		is the natural constraint for variable $v$.
		Box (f) also shows the solution matrix and Box (g) shows the induced type annotations,
    		now all predicted to be \lstinline{number}.
%

\subsection{An Outline of Probabilistic Type Inference}

We consider a dynamic language of untyped programs that is equipped with an existing deterministic type system,
that requires type annotations on identifiers.
Given a program $U$ plus a typing environment $\Gamma$ let $\Gamma \vdash U$ mean that the program $U$ is well-typed according to the (deterministic) type system, given types for identifiers provided by $\Gamma$.
Formally, a typing environment $\Gamma$ is a finite function
with domain $\{ x_v \mid v \in 1 \ldots V\}$, where $x_v$  is an identifier, and range $\{ l_\tau \mid \tau \in 1 \dots T \}$, where each $l_\tau$ is a literal type.
Given an untyped program $U$,
let \emph{probabilistic type inference} consist of these steps:
\begin{enumerate}
	\item We choose a finite universe of $T$ distinct literal types $\{ l_\tau \mid \tau \in 1 \dots T \}$.
	\item We compute a set $\{ x_v \mid v \in 1 \ldots V\}$ of a number $V$ of distinct identifiers in $U$ that need to be assigned types.
	\item \label{step:constraints} We extract a set of constraints from $U$.
	\item \label{step:optimise} By optimising these constraints, we construct a matrix $P$ with $V$ rows and $T$ columns,
	      such that each row is a probability vector
	      (a discrete distribution over the $T$ literal types).
	\item For each identifier $x_v$, we set type $t_v$ to the literal type $l_\tau$ which we compute from the $v$th probability vector (the one for identifier $x_v$).  In this work, we pick the column $\tau$ that has the maximum probability in $x_v$'s probability vector.
	\item The outcome is the environment $\Gamma = \{ x_v : t_v \mid v \in 1 \ldots V\}$.
\end{enumerate}

We say that probabilistic type inference is \emph{successful} if $\Gamma \vdash U$, that is, the untyped program $U$ is well-typed according to the deterministic type system.
Since several steps may involve approximation, the prediction $\Gamma$ may only be partially correct.
Still, given a known $\hat{\Gamma}$ such that $\hat{\Gamma} \vdash U$ we can measure how well $\Gamma$ has predicted the identifiers and types of $\hat{\Gamma}$.
A key idea is that there are two sorts of constraints in step~(\ref{step:constraints}): logical constraints and natural
constraints.

A \emph{logical constraint} is a formula $E$ that describes
necessary conditions for $U$ to be well-typed.
In principle, $E$ can be any formula such that if $\Gamma \vdash U,$
then $\Gamma$ satisfies $E$.
Thus, the logical constraints
do not need to uniquely determine $\Gamma$.
For this reason, a \emph{natural constraint}
encodes less-certain information about $\Gamma$,
for example, based on comments or names.
Just as we can conceptualise the logical
constraints as a function to the set of boolean values $\{0, 1\}$,
we can conceptualise the natural constraints as functions
that map $\Gamma$ to the set of probabilities $[0, 1]$, which can be interpreted
as a prediction of the probability that $\Gamma$ would
be successful. To combine these two constraints, we relax the boolean operations to continuous operators on $[0, 1]$.
Since we can conceptualise $E$ as a function
that maps $\Gamma$ to a boolean value $\{0, 1\}$,
we relax this function to map to $[0,1]$, using
a continuous interpretation of the semantics of $E$.
Similarly, we relax $\Gamma$ to a $V \times T$ matrix of probabilities.
Having done this,
we formalise type inference as a problem in
numerical optimisation, with the goal to find a relaxed type assignment
that satisfies as much as possible both sorts of constraints.
The result of this optimisation procedure is the
matrix $P$ of probabilities described in step~(\ref{step:optimise}).

\subsection{\textit{Logical} Constraints in Continuous Space}\label{ssec:logcon}

Logical constraints are extracted from our untyped input program $U$ using
standard program analysis techniques.  Here, we rely on a \emph{Constraint
	Generator} for this purpose. \cref{ssec:logprodts} describes its
realisation.  The generator takes into account a set of rules that the type system
enforces and produces a boolean type constraint for them.

In this work, we consider the following logical constraints.
\begin{definition}[\emph{Grammar of Logical Constraints}]\label{def:log-gram}
	A \emph{logical constraint} is an expression $E$ of the following form:
	\begin{equation}\label{eq:gram}
    \begin{IEEEeqnarraybox}[][c]{rl}
		E & \> ::= x_v \mathrel{is} l_\tau \\ 
		  & \>\mid{} \mathrel{not} E      \\
		  & \>\mid E \mathrel{and} E      \\
		  & \>\mid E \mathrel{or} E.
	\end{IEEEeqnarraybox}
\end{equation}

	Let $\mathcal{E}$ be the set of all logical constraints.
\end{definition}

Recall that a typing environment $\Gamma$ is a finite function
with domain $\{ x_v \mid v \in 1 \ldots V\}$, and range $\{ l_\tau \mid \tau \in 1 \dots T \}$.
The standard \emph{logical satisfaction} relation $\Gamma \models E$, is defined by induction on the structure of $E$, as follows.
The typing environment $\Gamma$ plays the role of a model for the formula $E$.

\begin{equation}\label{eq:logsat}
    \begin{IEEEeqnarraybox}[][c]{clcl}
        \Gamma & \> \models x_v \mathrel{is} l_\tau & \Leftrightarrow \> & \Gamma(x_v)=l_\tau\\
        \Gamma & \> \models{} \mathrel{not} E \> & \Leftrightarrow & \mbox{not $\Gamma \models E$}\\
        \Gamma & \> \models E_1 \mathrel{and} E_2\> & \Leftrightarrow & \mbox{$\Gamma \models E_1$ and $\Gamma \models E_2$} \\
        \Gamma & \> \models E_1 \mathrel{or} E_2 \> & \Leftrightarrow & \mbox{$\Gamma \models E_1$ or $\Gamma \models E_2$.}
    \end{IEEEeqnarraybox}
\end{equation}

\paragraph{Continuous Relaxation}
We explain how to specify a \emph{continuous relaxation} of the discrete logical semantics. Intuitively, the logical semantics defines a truth function that maps typing environments to $\{0, 1\}$; a continuous
relaxation extends the range of the truth function to $[0, 1].$
To see this, start with two auxiliary definitions:
\begin{itemize}
    \item We define $\Pi^{V \times T}$ to be the set
of all \emph{probability matrices} of size $V \times T$,
that is, matrices of the form $P = \begin{bmatrix} \bm{p}_1 & \ldots & \bm{p}_{V} \end{bmatrix}^\mathsf{T}$,
where each $\bm{p}_v = \begin{bmatrix} p_{v,1} & \ldots & p_{v,{T}} \end{bmatrix}^\mathsf{T}$
is a vector that defines a probability distribution over concrete types.
    \item We convert an environment $\Gamma$ into a $V \times T$ binary matrix $B(\Gamma)$ by setting $b_{v,\tau} = 1$ if $(x_v, l_\tau) \in \Gamma,$ and 0 otherwise.
    Each binary matrix is also a probability matrix: $B(\Gamma) \in \Pi^{V \times T}$.
\end{itemize}
Given a formula $E,$ we define a truth function $f_E: \{0, 1\}^{V \times T} \rightarrow \{0, 1\}$
that maps binary matrices to $\{0, 1\}$, namely, for all $\Gamma,$ we define $f_{E}(B(\Gamma)) = 1$ if and only if $\Gamma \models E$.
A \emph{relaxed semantics} is a continuous function
that always agrees with the logical semantics, that is,
a relaxed semantics is a function
$\tilde{f}_{E} : \Pi^{V \times T}  \rightarrow [0, 1]$
such that for all formulas $E$ and environments $\Gamma$,
$\tilde{f}_{E}(B(\Gamma)) = f_E(B(\Gamma))$.
Essentially, a relaxed semantics extends the domain and range
of $f_E$ to be continuous instead of discrete.


%
Our \emph{continuous semantics} (or relaxed semantics) $\qqpi{P}{E}$ is a function $\Pi^{V \times T} \times \mathcal{E} \rightarrow [0, 1]$,
defined by induction on the structure of $E$, as follows.
\begin{equation}\label{eq:logical}
    \begin{IEEEeqnarraybox}[][c]{rl}
        \qqpi{P}{x_v \mathrel{is} l_\tau} & \> = p_{v,\tau}\\
        \qqpi{P}{\mathrel{not} E} & \> = 1-\qqpi{P}{E}\\
        \qqpi{P}{E_1 \mathrel{and} E_2} & \> = \qqpi{P}{E_1} \cdot \qqpi{P}{E_2} \\
	\qqpi{P}{E_1 \mathrel{or} E_2} & \> = \qqpi{P}{E_1} + \qqpi{P}{E_2} - \qqpi{P}{E_1}\cdot\qqpi{P}{E_2}.
    \end{IEEEeqnarraybox}
\end{equation}
(In the actual implementation, we use logits instead of probabilities
for numerical stability, see \cref{app:appendix-logit}.)
Our semantics is based on standard many-valued interpretations of propositional logic formulas as described for instance by Hajek et al.~\cite{hajek98}; however, our atomic propositions $x_v \mathrel{is} l_\tau$ and their interpretation via the matrix $P$ are original.
To interpret conjunction, we use what is known as the \emph{product $t$-norm}, where the conjunction of two constraints is interpreted as the numeric product of their interpretations:
$\qqpi{P}{E_1 \mathrel{and} E_2} = \qqpi{P}{E_1} \cdot \qqpi{P}{E_2}$.
We make this choice because the numeric product is smooth and fits with our optimisation-based approach.
The product $t$-norm has already been used for obtaining relaxed logical semantics in machine learning, for example in~\cite{rocktaschel15}.
Other choices are possible as we discuss in \cref{ssec:softlogic}.


Our continuous semantics respects the duality between conjunction and disjunction:
\begin{lemma}[Duality]
For all $E_1$, $E_2$, and $P$:
\begin{enumerate}
    \item $\qqpi{P}{\mathrel{not}(E_1 \mathrel{and} E_2)} = \qqpi{P}{(\mathrel{not} E_1) \mathrel{or} (\mathrel{not} E_2)}$
    \item $\qqpi{P}{\mathrel{not}(E_1 \mathrel{or} E_2)} = \qqpi{P}{(\mathrel{not} E_1) \mathrel{and} (\mathrel{not} E_2)}$
\end{enumerate}
\end{lemma}

The following asserts essentially that the relaxed semantics is a continuous function that always agrees with the logical semantics.
\begin{theorem}[Relaxation]\label{thm:soft2hard}
For all E and $\Gamma$, $\qqpi{B(\Gamma)}{E} = 1 \Leftrightarrow \Gamma \models E$.
\end{theorem}
The proof is by induction on the structure of $E$;
the details are in~\cref{app:proofs}.

Our general formulation of a relaxed semantics was in terms of the functions $f$ and $\tilde{f}$,
where we defined $f$ by $f_{E}(B(\Gamma)) = 1$ if and only if $\Gamma \models E$.
We sought a relaxed semantics $\tilde{f}$ which we can now define by: $\tilde{f}_E(P) = \qqpi{P}{E}$.
Our desired equation $\tilde{f}_{E}(B(\Gamma)) = f_E(B(\Gamma))$,
for all formulas $E$ and environments $\Gamma$,
follows as a corollary of \cref{thm:soft2hard}.

Recall that in our setting, we know $E$ but
do not know $\Gamma$.
To address that, we observe that because the continuous semantics $\tilde{f}_E(P) = \qqpi{P}{E}$ is a function of $P$, we can optimise numerically
the function $\tilde{f}_E(P)$ with respect to $P \in \Pi^{V \times T}$.
If the optimisation is successful in finding an optimal value $P^*$ such that
$P^* \approx B(\Gamma)$ for some $\Gamma$ and that $\tilde{f}_E(P^*) = \qqpi{P^*}{E} = 1$,
then the theorem tells us that we have a typing environment $\Gamma$ that models $E$.




\subsection{\textit{Natural} Constraints via Machine Learning}\label{ssec:natcon}

A complementary source of information about types arises from statistical dependencies
in the source code of the program.  For example, names of variables provide
information about their types~\cite{xu16}, natural language in
method-level comments provide information about function types \cite{malik19},
and lexically nearby tokens provide information
about a variable's type~\cite{wei20,hellendoorn18}.
This information is indirect, and extremely difficult to formalise,
but we can still hope to exploit it by applying machine learning
to large corpora of source code.

Recently, the software engineering
community has adopted the term \emph{naturalness of source code} to refer to
the concept that programs have statistical regularities because
they are written by humans to be
understood by humans~\cite{hindle12}.
Following the idea that the naturalness in source code may be in part responsible
for the effectiveness of this information, we
refer generically to indirect, statistical
constraints about types as \emph{natural constraints}.
Because natural constraints are uncertain, they are naturally formalised
as probabilities.
A natural constraint is a mapping from a type variable to a vector
of probabilities
over possible types.
\begin{definition}[\emph{Natural Constraints}]\label{eq:natural}
	For each identifier $x_v$ in a program $U$,
	a \emph{natural constraint} is a probability vector $\bm{\mu}_v = [\mu_{v,1}, \ldots, \mu_{v,T}]^\mathsf{T}$.
	We aggregate the probability vectors of the learning model in a matrix
	defined as $\mathcal{M} = \begin{bmatrix} \bm{\mu}_1 & \ldots & \bm{\mu}_{V} \end{bmatrix}^\mathsf{T}$.
\end{definition}

In principle, natural constraints can be defined based on any property of $U$,
including names and comments.
In this paper, we consider a simple but practically effective example of
natural constraint, namely, a deep network that predicts the type
of a variable from the characters in its name.
We consider each variable identifier $x_v$ to be a character sequence $(c_{v1} \ldots c_{vN})$,
where each $c_{vi}$ is a character.
(This instantiation of the natural constraint is defined
only on types for identifiers that occur in the source code,
such as a function identifier or a parameter identifier.)
This is a classification problem, where the input is $x_v$,
and the output classes are the set of $T$ concrete types.
Ideally, the classifier would learn that identifier names that are lexically similar
tend to have similar types, and specifically which subsequences of the character names,
like \texttt{lst}, are highly predictive of the type, and which subsequences are less predictive.
One simple way to do so is to use a recurrent neural network (RNN).

For our purposes, an RNN is simply a function $(\bm{h}_{i-1}, z_i) \mapsto \bm{h}_{i}$
that maps a state vector $\bm{h}_{i-1} \in \mathbb{R}^H$
and an arbitrary input $z_i$ to an updated state vector $\bm{h}_{i}  \in \mathbb{R}^H$.
(The dimension $H$ is one of the hyperparameters of the model, which can be tuned
to obtain the best performance.)
The RNN has continuous parameters that are learned to fit a given data set,
but we elide these parameters to lighten the notation, because they are trained in a standard way.
We use a particular variant of an RNN called a
long-short term memory network (LSTM)~\cite{hochreiter97},
which has proven to be particularly effective both for natural language
and for source code~\cite{sundermeyer2012,melis17,white2015,dam16}.
We write the LSTM as $\text{LSTM}(\bm{h}_{i-1}, z_i)$.

With this background, we can describe the specific natural constraint that we use.
Given the name $x_v = (c_{v1} \ldots c_{vN}),$ we input each character $c_{vi}$ to the LSTM,
obtaining a final state vector $\bm{h}_N,$ which is then passed as input to a small
neural network that outputs the natural constraint $\bm{\mu}_v$.
That is, we define
\begin{subequations}\label{eq:lstm}
	\begin{align}
		\bm{h}_i   & = \text{LSTM}(\bm{h}_{i-1}, c_{vi}) \qquad i \in 1, \ldots, N \\
		\bm{\mu}_v & = F(\bm{h}_N), \label{eq:lstmb}
	\end{align}
\end{subequations}
where $F: \mathbb{R}^H \rightarrow \mathbb{R}^T$ is a simple neural network.
In our instantiation of this natural constraint, we choose $F$ to be a feedforward neural network with
no additional hidden layers, as defined in \eqref{eq:feedforward}.
We provide more details regarding the particular structure of our neural network in \cref{ssec:natprodts}.

This network structure is, by now, a fairly standard architectural motif in deep learning.
More sophisticated networks could certainly be employed, but are left to future work.

\subsection{Combining Logical and Natural Constraints to Form an Optimisation Problem}\label{ssec:optimisation}
Logical constraints pose challenges to the probabilistic world of
machine learning.  Neural networks cannot handle hard constraints explicitly and 
thus it is not straightforward how to incorporate the logical rules that they must follow.
Our way around that problem is to relax the logical constraints to numerical
space and combine them with the natural constraints through a continuous
optimisation problem.

Intuitively, we design the optimisation problem to be over
probability matrices $P \in \Pi^{V \times T}$; we wish to find
$P$ that is as close as possible to the natural constraints $\mathcal{M}$
subject to the logical constraints being satisfied.
A simple way to quantify the distance is via the \emph{Euclidean norm} $|| \cdot ||_2$ of a vector,  which is a convex function and thus well suited with our optimisation approach.
Hence, we obtain the constrained optimisation problem
\begin{equation}\label{eq:opt_naive}
    \begin{IEEEeqnarraybox}[][c]{c'l}
        \min_{P \in \mathbb{R}^{V \times T}} & \sum_v || \bm{p}_v - \bm{\mu}_v ||_2^2\\
        \text{s.t.} & p_{v\tau} \in [0, 1], \quad \forall v, \tau\\
        & \sum_{\tau=1}^T p_{v\tau} = 1, \quad \forall v \\
	& \qqpi{P}{E} = 1.
    \end{IEEEeqnarraybox}
\end{equation}

%
%

We use Mean Squared Error (MSE) here to quantify the performance of our fitting.
We could have used the Cross Entropy (CE), another common loss function.  The
MSE is a proper scoring rule~\cite{gneiting07}, meaning that smaller values
correspond to better matching of our optimisation variables with the logical
constraints.  We do not claim any particular advantage of MSE versus CE.

Instead of solving optimisation problem \eqref{eq:opt_naive}, we proceed to make some remarks that exploit its structure.
First, we reparameterise the problem to remove the probability constraints, by using the softmax function
\begin{equation}\label{eq:softmax}
	\sigma(\bm{x}) = \left[\frac{\exp\{x_1\}}{\sum_i \exp\{x_i\}}, \frac{\exp\{x_2\}}{\sum_i \exp\{x_i\}}, \cdots \right]^\mathsf{T},
\end{equation}
which maps real-valued vectors to probability vectors.
Our transformed problem takes the form
\begin{equation}
	\begin{aligned}\label{eq:opt_no_prob}
		\underset{Y \in \mathbb{R}^{V \times T}}{\mathrm{min}} & \quad
		\sum_v || \sigma(\bm{y}_v)^\mathsf{T} - \bm{\mu}_v ||_2^2                                                                         \\
		\text{s.t. } & \quad
		\qqpi{[\sigma(\bm{y}_1), \ldots, \sigma(\bm{y}_{V})]^\mathsf{T}}{E} -1 = 0.
	\end{aligned}
\end{equation}
It is easy to see that for $Y^*$ that minimises \eqref{eq:opt_no_prob}, then
$P^* = [\sigma(\bm{y}_1), \ldots, \allowbreak \sigma(\bm{y}_{V})]^\mathsf{T}$
minimises \eqref{eq:opt_naive}.
We remove the last constraint by introducing a multiplier $\lambda \in \mathbb{R}$, 
yielding the final form of our optimisation problem
\begin{equation}\label{eq:objective}
	\min_{Y, \lambda}
	\sum_v || \sigma(\bm{y}_v)^\mathsf{T} - \bm{\mu}_v ||_2^2
	- \lambda \big(\qqpi{[\sigma(\bm{y}_1), \ldots, \sigma(\bm{y}_{V})]^\mathsf{T}}{E} - 1\big).
\end{equation}
This corresponds to the Lagrange multiplier method~\cite{bertsekas82}.
According to the first-order necessary conditions, at a saddle point of the objective function of~\eqref{eq:objective} the following conditions should be satisfied
\begin{subequations}
\begin{align}
    \nabla_Y \mathcal{L}(Y, \lambda) & = 0\\
    \nabla_\lambda \mathcal{L}(Y, \lambda) & = 0,\label{eq:langrange_constraint}
\end{align}
\end{subequations}
where $\mathcal{L}(Y, \lambda)$ equals the objective of~\eqref{eq:objective}.
Equation~\eqref{eq:langrange_constraint} guarantees that at our optimisation's problem solution, the equality constraint in~\eqref{eq:opt_no_prob} is satisfied.
In general, the necessary conditions are concerned with a saddle point.
In our experience, we have not faced any issue converging to optimal solutions when starting with large initial values for $\lambda$.
Nevertheless, a more systematic study is required for removing such possibilities.

Finally, we note that by adding more terms to the combined objective function, we can 
extend the sources of information we are getting as inputs to other channels, such as dynamic analysis.






\section{\projectname: Predicting Types for TypeScript}\label{sec:prodts}

To evaluate our approach in a real-world scenario, we
implement an end-to-end application, called \projectname, which aims
to suggest missing types for TypeScript files. The goal of
\projectname's implementation 
is to serve as a proof of concept for our general framework.
Thus, we acknowledge that our mechanisms to generate logical and natural constraints are both pragmatic under-approximations, and could in further work be replaced by more sophisticated mechanisms.

Regarding natural constraints, as every learning model outputs a probability vector over types, we can extend our method to include natural constraints generated by LambdaNet~\cite{wei20}, DeepTyper~\cite{hellendoorn18} or indeed any other deep learning approach that offers the same kind of output.

Regarding logical constraints, unfortunately the TypeScript typechecker does not produce explicit logical formulas to constrain unknown types.
Therefore, to build our prototype we generate logical constraints
by relying on a mode of operation where the compiler infer some types from usage on TypeScript code.
By doing that we obtain less information than we would if a full typechecker were to emit logical constraints, but nonetheless show state-of-the-art performance.

Even with the limitations of these two implementation choices for extracting logical and natural constraints, we show that \projectname outperforms JSNice, DeepTyper, and LambdaNet.

\subsection{Background: TypeScript's Type System}\label{ssec:intro-typescript}

TypeScript~\cite{typescript} is a typed superset of
JavaScript designed for developing large-scale, stable applications.
TypeScript's compiler (\lstinline+tsc+) typechecks TypeScript programs then emits plain JavaScript,
to leverage the fact that JavaScript is the only cross-platform
language that runs in any browser, any host, and any OS.

Structural type systems consider record types (classes), whose fields or members have the same names and types, to be equal.
TypeScript supports a
structural type system because it permits TypeScript to handle many JavaScript idioms that depend on dynamic typing.
One of the main goals of TypeScript's designers is to support idiomatic
JavaScript to
provide a smooth transition from JavaScript to TypeScript.
Therefore, TypeScript's type system is deliberately
unsound~\cite{understandtypescript}.  It is an optional type system, whose
annotations can be omitted and have no effect on runtime. As, TypeScript erases 
them when transpiling to JavaScript~\cite{understandtypescript}.
TypeScript's type system defaults to assigning its \texttt{any} type to unannotated parameters, and methods or properties signatures.  

\subsection{\textit{Logical} Constraints for TypeScript}\label{ssec:logprodts}
We cannot rely on the TypeScript compiler~\cite{typescript} \lstinline+tsc+ directly to generate the logical constraints of \cref{ssec:logcon}, because \lstinline+tsc+ does not construct logical formulas explicitly during typechecking.
Still, we can rely on a mode of operation where the compiler infers types on TypeScript code with no types annotations. To ensure that no types 
are present to the input files the first step 
of our process it to parse the gold ones and remove all type annotations, then we continue with generating the constraints.

Specifically, to generate logical constraints on argument types, we build on a command-line tool, named TypeStat~\cite{typestat}, that calls \lstinline+tsc+ to predict type hints for function arguments, usually to provide codefixes within a development environment.
When the predicted type for an argument identifier $x_v$ is a literal $l_\tau$ within our universe we emit the constraint $x_v \mathrel{is} l_\tau$.
When the predicted type for $x_v$ is a union type $(l_1 \mid \dots \mid l_n)$ of literals, we emit
the disjunction $((x_v \mathrel{is} l_1) \mathrel{or} \dots \mathrel{or} (x_v \mathrel{is} l_n))$.

Inferring return types is easier for the compiler than inferring argument types.
We can harvest return types as $\mathrel{is}$ type constraints by calling \lstinline+tsc+ directly.
Since inference of the return types is more straightforward than argument types, we see in~\cref{tab:typeacc1} that the \textit{Logical} phase of our method works better for return types than parameters.

Overall, for each function or group of functions in a file, we return a conjunction of the logical constraints generated for parameter and return types, as described above. We note
that our logical constraints include propositional logic, and therefore are able to express a wide range of interesting type constraints~\cite{odersky99,pottier05}.


\subsection{\textit{Natural} Constraints for TypeScript}\label{ssec:natprodts}

To learn naming conventions over types we use a Char-Level LSTM which predicts for any identifier a probability vector over all the available types. Our model is trained on
\textit{(id, type)} pairs with the goal to learn naming conventions
for identifiers, treated as sequences of characters.
The main intuition behind this choice is that
developers commonly use multiple abbreviations for the same word and
this family of abbreviations shares a type.
A Char-Level model is well-suited to predict the type for any identifier in
an abbreviation families.

\paragraph{Data Set for the LSTM}\label{par:datase} Following the work of 
~\cite{wei20} and~\cite{hellendoorn18},
to train our model we use as dataset the 300 
most starred Typescript projects from Github, 
containing between 500 to 10,000 lines of code.
Our dataset was randomly split by project into 70\% training data, 10\% validation data and 20\% test data. \cref{fig:Char-Level} shows a summary of the pipeline used to train our model, for specific implementation details of the LSTM refer to \cref{app:neural}.

\paragraph{Prediction Space}
We define our type vocabulary to consist of top-100 most common default library types in our training set. As Wei et al.~\cite{wei20} report, this prediction space covers 98\% of the non-any annotations for the training set.
Handling a larger set of types is straightforward, but we decided instead to work with the same set of library types used by \cite{hellendoorn18} and \cite{wei20}, to be consistent in our comparisons.
Conforming to the practice of prior work in this space \cite{wei20,hellendoorn18,xu16,raychev15}, 
we consider all different polymorphic type arguments to be \texttt{any}; for example a \texttt{Promise<boolean>} type corresponds to \texttt{Promise<any>}.
Accordingly higher-order functions correspond to the type \texttt{Function}.

\begin{figure}[!t]
    \centering
    \def\svgwidth{\linewidth}
    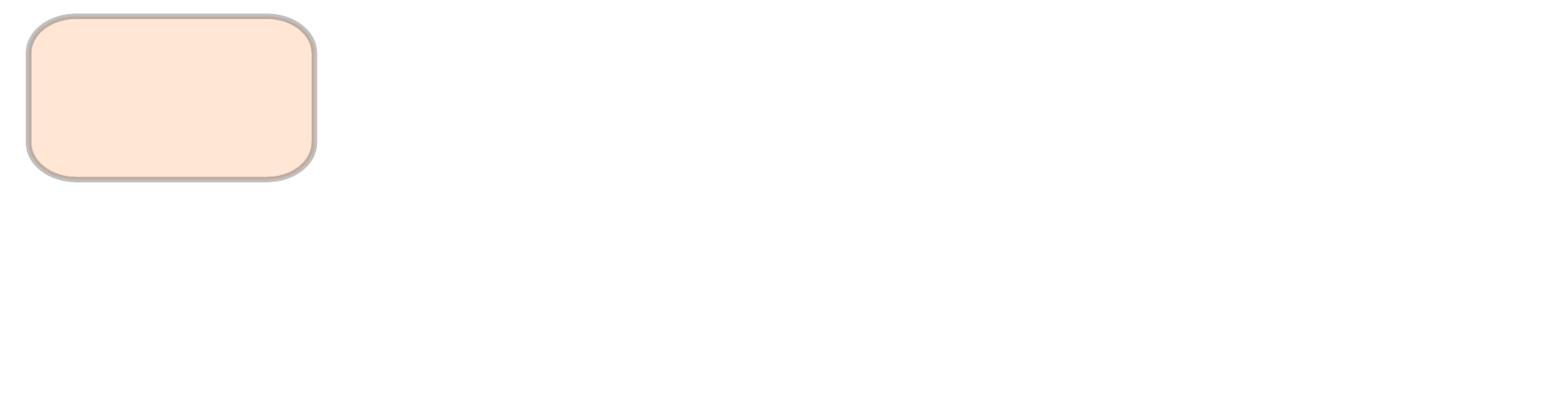
  \caption{Pipeline of learning naming conventions with 
  a Char-Level \textit{LSTM}, represented by a probability vector for each identifier.
  }\label{fig:Char-Level}
\end{figure}

\paragraph{Implementation Details}

Regarding the implementation details of the LSTM network, for the $F$ in \eqref{eq:lstmb},
we use a feedforward neural network
\begin{equation}
	F(\bm{h}) = \log\left( \sigma\left(\bm{h}A^T + b \right) \right),\label{eq:feedforward}
\end{equation}
where the $\log$ function is applied componentwise,
and $A$ and $b$ are learnable weights and bias.
The softmax function \eqref{eq:softmax} corresponds to the last layer of our neural network
and essentially maps the values of the previous layer to $[0, 1]$,
while the sum of all values is $1$ as expected for a probability vector as already explained.
We work in log space to help numerical stability since computing \eqref{eq:softmax} directly can be problematic.
As a result, $F$ outputs values in $[-\infty, 0]$.

We train the model by supplying sets of variable identifiers together with their known types,
and minimising a loss function.
Our loss function is the negative log likelihood function---conveniently combined with our log output---defined as
\begin{equation}
	L(\bm{y}) = -\sum_i log(\bm{y}_i).
\end{equation}
Essentially, we select, during training, the element that corresponds
to the correct label from the output $F$
and sum all the values of the correct labels for the entire training set.

\subsection{Combining \textit{Logical} and \textit{Natural} Constraints}
In our framework both solving the relaxed logical constraints alone, and combining them with the natural constraints correspond to an optimisation 
problem as described in~\eqref{eq:logical} and
~\eqref{eq:objective} respectively.
To find the minimum of the generated function we use 
\textsc{RMSprop}~\cite{tieleman2014};
an alternative to stochastic gradient descent~\cite{robbins51}, with an adaptive learning rate. Both the code for the deep
learning and the combined optimisation part is written in
PyTorch~\cite{paszke2017}.







\section{Evaluation of \projectname}
This section opens with \cref{sec:acc}, which defines the prediction and query spaces, explains the experimental setup and establishes the performance measure we use throughout the section.
\projectname is built to combine and exploit both logical and natural constraints, so \cref{ssec:ablation} quantifies their separate contributions and demonstrates their synergy.
It then compares and analyses the performance of \projectname with that of LambdaNet and DeepTyper.
It closes by comparing \projectname with JSNice, the pioneering work in probabilistic type inference.


\paragraph{Type Prediction Accuracy}\label{sec:acc}

For a statically typed language, let a declaration slot be a point in a program text where the grammar permits annotating an identifier with its type. The set of nodes that we predict types for includes variables (\textit{VAR}), parameters (\textit{PAR}), properties declarations and signatures (\textit{PROP}), return types for function declarations (\textit{FUN}) and method declarations (\textit{METH}).

Predicting types is a multi-class prediction task: at each annotation slot, we ask the predictor to propose a type from our type vocabulary, $\mathcal{T} = \{ l_\tau \mid \tau \in 1 \dots T \}$. We note that $\mathcal{T}$ does not include the gradual \texttt{any} type or Out-Of-Vocabulary \textit{OOV} token.  Concretely,
\cref{ssec:natprodts} defines $\mathcal{T}$.
Even if we regard some types as \textit{OOV} for our evaluation, we do allow all return types inferred by the \textit{Logical} part of our method to be present at the final output as they may be useful for the programmer.

\autoref{fig:type:prediction:query} shows the query space for types in an optional type setting.
TypeScript itself defines built-in types.  Library types are those types defined by the libraries a project imports and project types, sometimes called user types in the literature, are those types a project defines itself. 
In general, developer annotations are those slots that a developer is likely to annotate; in training data, we under-approximate these slots by the slots that a developer did annotate. 
The compiler inferable slots are those slots for which an optional compiler can infer a type given the developer annotations.  
Developers annotate some slots to document and clarify the code, to aid navigation and completion, and so that the compiler can infer types for other slots.  
Developer annotated slots are special for two reasons:  they provide a natural source of labelled training data and, since developers went to the effort of annotating them, relieving developers of the burden of doing so is clearly useful.


Let $\mathit{TP_i}$ be the number of times that a probabilistic type predictor \emph{correctly} labelled a slot with the $i^\text{th}$ type in $\mathcal{T}$;  let $\mathit{FP_i}$ be the number of times that the type predictor \emph{incorrectly} labelled a slot with the $i^\text{th}$ type.  Our ground truth for determining whether a prediction is correct is the set of developer annotated slots and the set of types that the compiler can infer; in our test set, we call this set the \texttt{gold} file. 
This is a working assumption in the sense that some files may contain errors~\cite{williams17}.

We report performance using Top-1 Accuracy~\cite{manning}, defined as 

\begin{align}\label{eq:acc}
    \frac{ \sum_i^T\mathit{TP_i}}{ \sum_i^T (\mathit{TP_i} + \mathit{FP_i})},
\end{align}
where $T=|\mathcal{T}|$. 
We perform our evaluation in all declaration and signatures annotation slots available in a TypeScript file.
\paragraph{Test Set} We evaluate \projectname on the same $20\%$ test set kept for the evaluation of our Char-Level model. Our data set consists of 2074 files from 60 GitHub packages with $\sim$10000 declaration annotation slots in total.
\begin{figure}[t]
\begin{lstlisting}[language=JavaScript]
    import { NotificationService, SimpleNotificationsModule }
        from 'angular2-notifications';

    class AppComponent {
        constructor(
         public service: Object
        ) {}
        open() {
            this.service.create('bla', 'blu');
        }
    }
\end{lstlisting}
   \vspace{-10mm}
    \caption{A snippet from \texttt{flauc$\_$angular2-notifications} project that fails to predict a correct type. On ln. 6 we predict that \lstinline{service} is an \lstinline{Object} as we don't have any \textit{Logical} constraint we base our decision based on the LSTM's prediction but the compiler complains that \lstinline{Property `create' does not exist on type `Object'}.The correct type is \lstinline{NotificationService},
    which has been imported.
}\label{fig:code-typecheck}
\end{figure}
\paragraph{Type Checking Results}
Our approach does not guarantee that
the resulting code always type check. Our only guarantee is that we respect the logical constraints that we generate from the code. There are cases where these constraints underdetermine the problem and the natural constraints help to mitigate that, but not necessarily successfully. Furthermore, optional typing checks for local type consistency, not the whole program guarantee of traditional type checking and that may lead to unspecified types from \lstinline{imports}. For a specific example, where this occurs and leads to a wrong prediction see~\cref{fig:code-typecheck}. Nevertheless, bridging two different sources of information, logical and natural, is not a trivial process and our aim in this paper is to show a principled way to tackle this problem by solving a corresponding numerical optimisation problem.


\subsection{Ablation Analysis:  Leveraging Both \emph{Logical} and \emph{Natural} Constraints}\label{ssec:ablation}

\begin{table}[t]
	\centering
	\caption{Ablation analysis of \projectname, the cells report Top-1 Accuracy(\%); \textit{FUN} and \textit{METH} refers to return types of functions and methods respectively, \textit{PAR}
	represents parameters, \textit{PROP} represents properties, and \textit{VAR} variables.
	Test set of 2074 files from 60 GitHub packages on $\sim$10000 annotation slots~\cite{wei20}.
	\\}\label{tab:typeacc1}

	\begin{tabular}{ccccccc}
		\toprule
		Tool  & $\textit{FUN}$ & 
		$\textit{METH}$ & 	$\textit{PAR}$ & 	$\textit{PROP}$& 	$\textit{VAR}$& $\textit{ALL}$ \\
		
		\midrule
		\textit{Logical}      & 0.69                                  & 0.71 & 0.06 & 0.08 & 0.42                                 & 0.23                               \\
		\textit{Natural}      & 0.42                                    & 0.52                                 & 0.82 & 0.40 & 0.37 & 0.69                                  \\
		\projectname  & 0.69         & 0.71         & 0.85 & 0.43 & 0.56 & 0.76\\
		\bottomrule
	\end{tabular}
\end{table}

\projectname has two phases --- \textit{Logical} and \textit{Natural} --- and combines them.  To evaluate the effect of each stage, \cref{tab:typeacc1} reports the Top-1 Accuracy of each stage.
\cref{tab:typeacc1}'s columns define disjoint sets of slots.
Regarding the \textit{Logical} approach, the results are significantly better for function and method return types
than parameter types.
This happens because the \projectname harvests a richer set of constraints for return types.
For parameter types, the \textit{Logical} stage  (\cref{ssec:logprodts})
most of the times fails to infer the corresponding type 
from its usage in code. An example where no hard constraint is available and the Char-Level fails to find a correct type is shown below in \cref{fig:code-typecheck}.


For the \textit{Natural} phase, the results swap, the prediction accuracy for the parameters' is almost double than the one for the return types. Our assumption is that
this is a result of largely using the same parameters ids over different projects, for example \lstinline{path}, than using the same function or methods ids. Nevertheless, the results from our Char-Level model indicate that just the naming of a variable carries a
lot of information about its actual type;   
Overall, \cref{tab:typeacc1} show that the \textit{Logical} and \textit{Natural} phases complements each other and thus their combination in \projectname greatly improves our type inference capabilities.

We note here that the concrete implementations of both phases of \projectname serve 
as a proof of concept to our theoretical approach
described in \cref{sec:framework}. Therefore,
 there is space of further improving the overall
 result by improving each of the two components with more sophisticated techniques.

\subsection{Comparison with DeepTyper and  LambdaNet}\label{ssec:typesubproblem}
\begin{table}[!t]
	\centering
		\caption{Top-1 Accuracy(\%) for DeepTyper, LambdaNet and \projectname.
		         (Same test set as Table~\ref{tab:typeacc1}.)
	\\}\label{tab:typeacc2}
		    \vspace{0.45cm}

	\begin{tabular}{ccccccc}
		\toprule
		Tool   &  Acc(\%) \\
		\midrule
		DeepTyper  & 0.61 \\
		LambdaNet  & 0.72 \\
		\projectname & 0.76\\
		\bottomrule
	\end{tabular}
\end{table}

\paragraph{DeepTyper}
There are two main differences between DeepTyper 
and \projectname that we need to address. Firstly, DeepTyper considers a much larger prediction space of $T=1100$ types, including many project types. Thus, to measure the accuracy, we restrict the prediction space to 100 library types (a subset of DeepTyper's vocabulary, exactly those LambdaNet's authors~\cite{wei20} chose when comparing DeepTyper to LambdaNet. We note that our \textit{Logical} phase is able to infer types that are out of this small prediction space, but for a fair evaluation we choose to treat them as \textit{OOV}.

Secondly, DeepTyper predicts a type, sometimes different, for each occurrence of an identifier, while we predict types only for declaration slots. To address this, we compare \projectname with a DeepTyper variant proposed by Wei et al.~\cite{wei20}.
This variant makes a single prediction for each identifier, by averaging over the RRN internal states for a particular identifier before the actual prediction.
The DeepTyper results we report are for this variant, retrained over the vocabulary of 100 library types specified above, using Top-1 Accuracy.

\cref{tab:typeacc2} summarises the results of our comparisons. We conjecture that \projectname outperforms DeepTyper mainly because \projectname's logical constraints define
a wider, lexically independent, prediction context than DeepTyper.
Perhaps taking into account only information in the vicinity as DeepTyper does can be problematic;
for instance, function definitions may be placed relatively far away from their calls and hence the context, that DeepTyper learns, is not very informative.
\begin{figure}[t]
    \centering
    \begin{minipage}[t]{.38\textwidth}
        \raggedright
\begin{lstlisting}[language=JavaScript,label={lst:gold}, numbers=left]
  function f1(
    x: boolean,
    z: Window,
    y: Event
  ): number {
      x = true;
      z = window;
      y = Event.prototype;
      return 1;
  }
\end{lstlisting}
   \vspace{3mm}
    \subcaption{Gold TypeScript file.}
    \end{minipage}
    \begin{minipage}[t]{.57\textwidth}
        \raggedleft
\begin{lstlisting}[language=Fake,label={lst:lncode}] 
function !$\mathscr{P}1\{\text{f}1\}($!
!$\mathscr{P}2\{\text{x}\}\textcolor{mygreen}{: \mathscr{L}}$\textcolor{mygreen}{[ty]\{Boolean\}}!,
!$\mathscr{P}3\{\text{z}\}\textcolor{Maroon}{: (\mathscr{L}}$\textcolor{Maroon}{[ty]\{String\} $\ne$\ }$\textcolor{Maroon}{\mathscr{L}}$\textcolor{Maroon}{[ty]\{Window\})}!,
!$\mathscr{P}4\{\text{y}\}\textcolor{Maroon}{: (\mathscr{L}}$\textcolor{Maroon}{[ty]\{Number\} $\ne$\ }$\textcolor{Maroon}{\mathscr{L}}$\textcolor{Maroon}{[ty]\{Event\})}!
): (!$\mathscr{P}5$$\textcolor{mygreen}{: (\mathscr{L}}$\textcolor{mygreen}{[ty]\{Number\}}!) {
    !$\mathscr{P}2\{\text{x}\} \leftarrow \mathscr{L}\{\text{Boolean}\}$!;
    !$\mathscr{P}3\{\text{z}\} \leftarrow \mathscr{L}\{\text{window}\}$!;
    !$\mathscr{P}4\{\text{y}\} \leftarrow \mathscr{L}\{\text{Event}\}\text{.prototype}$!;
    return !$\mathscr{P}5 = \mathscr{L}\{\text{Number}\}$!
}
\end{lstlisting}
    \subcaption{LambdaNet output.}
    \end{minipage}
    \caption{Minimal example showing 2 cases where LambdaNet gives incorrect predictions.}\label{fig:LambdaNet}
\end{figure}

\paragraph{LambdaNet} The comparison with LambdaNet is straightforward because they provide a pretrained model trained on the same dataset and for the same set of types as ours. 
\cref{tab:typeacc2} shows that our accuracy is on par with LambdaNet's, and indeed somewhat better.
We think that a reason for the small difference between LambdaNet's perfomance in all declaration slots and the reported performance in \cite{wei20} is due to the fact that LambdaNet fails to learn types that are 
inferable for the compiler but yet not apparent in 
the training data. We have found examples where this occurs. \cref{fig:LambdaNet} shows two examples.  
The parameters on lines 3 and 4 actually have type \texttt{Window} and \texttt{Event}, as you can see in \cref{lst:gold}, which contains the developer-annotated ground truth.
\cref{lst:lncode}, the figure on the right, shows that LambdaNet mispredicts their types as \texttt{String} and \texttt{Number}.
We conjecture that the misprediction is because of data sparsity.
LambdaNet correctly predicts the type of the first parameter because uses of \texttt{boolean} are relatively common in the training data, while uses of \texttt{Window} and \texttt{Event} are not, so the assignments on lines 7 and 8 provide too little signal for LambdaNet to pick up.
\projectname, in contrast, correctly predicts all three parameter types.
\projectname succeeds here because the assignments on lines 7 and 8 generate hard logical constraints that \projectname incorporates, \emph{at test time}, into its optimisation search for a satisfying type environment.
These examples may explain the difference between LambdaNet's and \projectname's prediction accuracy. 
This difference in performance between the two approaches will crop up whenever the training data lacks sufficient number of examples of a particular logical relation. 

Finally, we note 
that we see LambdaNet and indeed any learning approach is complementary to our work. In theory, it is 
straightforward to treat them as an instantiation of the \textit{Natural} phase, as a probability distribution over types described in our theoretical framework (\cref{sec:framework}).

\subsection{Comparison with JSNice}
\begin{table}[!t]
	\centering
	\caption{Top-1 Accuracy for JSNice and \projectname.
		     (Same test set as Table~\ref{tab:typeacc1}.)
	}~\label{tab:typeacc3}
	 \vspace{0.45cm}

	\begin{tabular}{ccccccc}
		\toprule
		Tool   & Acc \\
		\midrule
		JSNice     & 0.45\\
		\projectname & 0.74\\
		\bottomrule
	\end{tabular}
\end{table}
Only portions of the JSNice~\cite{raychev15} system have been made open source. 
The portion of the implementation that the authors have made public is not sufficient to retrain the JSNice models. 
Instead, we follow the approach of DeepTyper~\cite{hellendoorn18} and LambdaNet~\cite{wei20} and manually compare JSNice and \projectname over a smaller dataset. 
As JSNice targets JavaScript, it cannot predict types for classes or interfaces, so we report accuracy on top-level functions return types and parameters randomly sampled from our test set. The prediction
space for this comparison is restricted to JavaScript primitive types.

As \cref{tab:typeacc3} shows, \projectname outperforms JSNice. We conjecture that this is because JSNice exploits the relation paths between types up to a shallow depth, and thus it may not
capture some typing relevant element dependencies. 
In contrast, our logical constraints can leverage information of a more expansive context.
Additionally, \projectname's grouping together of names that share a type despite minor variations in their names could be proven useful for learning techniques.

\section{Related Work}\label{sec:related}

\projectname is a new form of probabilistic type inference that optimises over
both logical and natural constraints.  Related work spans classical, deterministic
type inference, soft logic for the relaxation of the constraints, and earlier machine learning approaches.

\subsection{Classical Type Inference}


Rich type inference mitigates the cost
of explicitly annotating types. This feature is an
inherent trait of strongly, statically-typed, functional languages (like Haskell or ML).


Dynamic languages have also started to pay more attention to typings. Several
JavaScript extensions, like Closure Compiler~\cite{closure}, Flow~\cite{flow} and
TypeScript (See \cref{ssec:intro-typescript}) are all focusing on enabling
sorts of static type checking for JavaScript.
%
%
However, these
extensions often fail to scale to realistic programs that make use of dynamic
evaluation and complex libraries, for example, jQuery, which cannot be analysed
precisely~\cite{jensen09}.
There are similar extensions for other popular scripting languages,
like~\cite{mypy}, an optional static type checker for Python,
or RuboCop~\cite{rubycop}, which serves as a static analyzer for Ruby by enforcing many of the guidelines
outlined in the community Ruby Style Guide~\cite{rubystyle}.

The quest for more modular and extensible static analysis techniques has
resulted in the development of richer type systems.
Refinement types, that is, subsets of types that satisfy a logical predicate (like Boolean expression),
constrain the set of values described by the type and hence allow the use of
modern logic solvers (such as SAT and SMT engines) to extend the
scope of invariants that can be statically verified.
An implementation of this concept comes with Logically Qualified Data Types,
abbreviated to Liquid Types.
DSOLVE is an early application of liquid type inference in OCAML~\cite{liquid}.
A type-checking algorithm, which relies on an SMT solver
to compute subtyping efficiently for a core, first-order functional language
enhanced with refinement types~\cite{semanticSMT}, provides a different
approach.
LiquidHaskell~\cite{refHaskell} is a static verifier of
Haskell based on Liquid Types via SMT and predicate
abstraction.
DependentJS~\cite{dependentJS} incorporates dependent types into JavaScript. We note also, that HM(X) is a family of constraint-based type
systems~\cite{odersky99,pottier05}, that fits our formulation of the logical constraints.

A line of research closely related to our work concerns
the specific problem of predicting a TypeScript declaration file for an underlying JavaScript library. Writing and maintaining a declaration file is a non-trivial process. Both TSCHECK~\cite{feldthaus14} and TSTEST~\cite{kristensen17} demonstrates the difficult of the task, by detecting numerous errors in the declaration files of most of the libraries they have checked. \cite{tstools17} created the TSINFER and TSEVOLVE tools to that to assist the 
programmer for creating and maintaining TypeScript declaration files from JavaScript files. These tools are based on a combination of a static and dynamic analysis that uses a recorded snapshot of a concretely initialised library to generate TypeScript declaration files from JavaScript libraries. Dynamic analysis for TypeScript could serve as a different
source of information in our type inference framework. We
could capture results from dynamic analysis by adding a new term to our objective~\eqref{eq:objective}.




\subsection{Probabilistic Type Inference}\label{sec:ml:over:source}

Although the interdisciplinary field between machine learning and programming
languages is still young, complete reviews of this area are
already available.
A detailed description of the area can be found in~\cite{vechev16}.
While~\cite{threepillars} is a position paper, which examines this research area by categorising
the challenges involved in three main, overlapping pillars: intention, invention, and adaptation.
An extensively survey work that probabilistically models source code via a learning component and complex representations of the underlying code is given in~\cite{allamanis17}.

A sub-field of this emerging area applies machine learning in probabilistic graphical models to infer semantic properties of programs, such as types.
The first example of this class
of approach was
JSNice~\cite{raychev15}, which uses probabilistic graphical models to infer types (and deobfuscate names) for JavaScript files. 
They use conditional random fields (CRF)~\cite{sutton12}, to encode variable relations between types, but not type constraints. 
\projectname differs from JSNice in incorporating logical type constraints and reformulating  probabilistic type inference as an optimisation problem.
LambdaNet~\cite{xu16} use a different  graphical model to statistically infer types for Python. 
Their method trains a
classifier for each project that predicts the
type of each variable from its name.
The classifier's predictions for each variable are combined with semantic constraints using a kind of graphical model called a \emph{factor graph}~\cite{yedidia2003}, which is
closely related to CRFs.
To build their factor graph, they leverage type hints
derived from data flow, attribute accesses, subtyping, and naming conventions. Compared to our work, \cite{xu16} requires
heuristically chosen weights in the factors that integrate naming and semantic constraints; these heuristics may need to be tuned separately for each new kind of semantic constraint that is added to the model. In contrast, our method relies on an optimisation that integrates semantic constraints
and naming constraints in a principled way.
An important advantage of our approach
is that we can introduce type constraints from any standard type inference engine by automatically relaxing them.

Most recent works have used deep learning approaches to tackle
the problem of probabilistic type inference. Hellendoorn et al. were the first to do so, by building 
DeepTyper, a tool that infers types for partially typed TypeScript code~\cite{hellendoorn18}. DeepTyper uses a bidirectional neural network that leverages local lexical information to make type predictions for every identifier but otherwise ignores type constraints. 
A strength of DeepTyper is that it can handle a very large type vocabulary that spans both library and user-defined types.

NL2Type, a tool by~\cite{malik19}, also takes
a deep learning approach to the task, using JSDoc comments as an additional type hint.
Our approach differs from DeepTyper and NL2Type in  incorporating logical type constraints.
Our reformulation of probabilistic type inference as an optimisation problem that incorporates explicitly logical constraints separates us from these two approaches. Williams et al. present a different application of
combining logical constraints and machine learning for inferring units in spreadsheets~\cite{williams20}.

LambdaNet~\cite{wei20} was the first to apply a \emph{graph neural network} (GNN) model \cite{Gilmer2017-qd,allamanis17a} to the probabilistic type inference task.
LambdaNet targets TypeScript and uses static analysis to build its GNN from the training data. 
This GNN combines logical and contextual (which subsumes \projectname's natural) constraints.
They are the first to be able to predict \emph{unseen} user-defined types by using a pointer-network-like architecture~\cite{vinyals15,Allamanis2016-su} to predict over an open vocabulary.
Typilus~\cite{allamanis20} is a GNN that employs one-shot learning to predict an open vocabulary of types, including rare and unseen user-defined types, for Python. 
Interestingly, both of these GNN models use an iterative
computation called message-passing to compute predictions,
which is closely related to the sum-product algorithm
used for inference in~\cite{xu16}, which also relies on message passing.
The main difference between our method and these two GNN-based approaches is that we do not learn the logical constraints but rather extracts and enforces them at test time.
Incorporating previously known, hard constraints into a learning
model is an effective way to improve its overall performance, but it does not imply that the result will respect them. What our approach does instead, is to offer a principled way to explicitly impose the logical constraints while constructively, 
and only at the places where is needed, absorbing information from the natural channel. In that sense, our work is complementary to each prior work above, as the described tools output
a probability vector over types, which our approach can take as its input to the
\textit{Natural} part of our combined optimisation equation \eqref{eq:objective}.

TypeWriter~\cite{pradel19} applies probabilistic type checking to infer argument and result types for Python programs.
To solve the problem that inferred types may not correctly typecheck, TypeWriter invokes a standard gradual type checker to validate different combinations of types suggested by probabilistic inference, so to suggest types that will correctly typecheck.
Adding a second phase to \projectname, in the spirit of TypeWriter, would guarantee that the results of \projectname would typecheck.

\section{Conclusion}\label{sec:conclusion}

This paper addresses the lack of rich type inference process for dynamically typed languages.
To tackle this, we combine \textit{Logical} constraints, deterministic information from a type system, with \textit{Natural} constraints, uncertain information about types, learnt by machine learning techniques, while focusing on the satisfaction of the typing rules dictated by the language.

A core aim of our method is to guide the predictions from the learning procedure to respect the logical constraints.
We achieve this by relaxing the logical type inference problem into a continuous space.
This allows us to constructively combine the natural and logical part in a single optimisation problem with guaranteed constraint satisfaction. 
Our framework is extensible: it can incorporate information from arbitrary models into its natural part and type constraints generated by traditional deterministic type inference systems.

We evaluate our framework by implementing \projectname, a tool that predicts types for TypeScript.
Our experiments show that \projectname achieves an accuracy of 76\% for Top-1 prediction on the standard test set~\cite{wei20} for probabilistic type inference on TypeScript, an improvement of 4\% on previous work.
An ablation study shows the improved performance arises from our new principled combination of both \textit{Logical} and \textit{Natural} constraints at test time.

\bibliographystyle{splncs04}
\balance
\bibliography{references}


\appendix
\section{Appendix: Continuous Relaxation in the Logit Space}\label{app:appendix-logit}




In \cref{ssec:logcon}, we present the continuous interpretation based on probabilities.
As already mentioned, in the actual implementation we use logit instead for numerical stability.
The \emph{logit} of a probability is the logarithm of the odds ratio.
It is defined as the inverse of the softmax function; that is, an element of a probability vector $p \in [0,1]$ corresponds to
\begin{equation*}
	\pi = \log \frac{p}{1 - p}.
\end{equation*}
It allows us to map probability values from $\left[ 0, 1 \right]$ to $\left[ -\infty, \infty \right]$.

Given the matrix~$\mathcal{L}$, which corresponds to the logit of the matrix~$P$ in \cref{ssec:logcon},
we have that $\text{log}(\qqpi{P}{E})  = \qqpi{\mathcal{L}}{E}$.
we interpret an expression~$E$ as a number ~$\qqpi{P}{E} \in \mathbb{R}$ as  follows:
\begin{align*}
	\qqpi{\mathcal{L}}{x_v \mathrel{is} l_\tau} & = \pi_{v,\tau}                                                \\ \label{eq:logits}
	\qqpi{\mathcal{L}}{\mathrel{not} E}         & = \log(1-\text{exp}(\qqpi{\mathcal{L}}{E})                \\
	\qqpi{\mathcal{L}}{E_1 \mathrel{and} E_2}   & = \qqpi{\mathcal{L}}{E_1} \mathrel{+} \qqpi{\mathcal{L}}{E_2} \\
	\qqpi{\mathcal{L}}{E_1 \mathrel{or} E_2}    & = \text{log}\big(
	\text{exp}(\qqpi{\mathcal{L}}{E_1}) + \text{exp}(\qqpi{\mathcal{L}}{E_2}) - \text{exp}(\qqpi{\mathcal{L}}{E_1} + \qqpi{\mathcal{L}}{E_2})\big).
\end{align*}
The sigmoid function is defined as
\begin{equation*}
	\text{sigmoid}(a) = \frac{\exp\{a\}}{1 + \exp\{a\}},
\end{equation*}
while the LogSumExp function is defined as
\begin{equation*}
	\text{LogSumExp}(\bm{x}) = \log\left( \sum_i \exp\{x_i\} \right).
\end{equation*}

\section{Appendix: Formal Proofs}\label{app:proofs}
\subsection{Proofs for Logical Constraints}

\begin{lemma} \label{lem:binary}
    For all $E$ and $\Gamma$, $\qqpi{B(\Gamma)}{E} \in \{0,1\}$.
\end{lemma}

\begin{proof} 
    By structural induction on the continuous semantics.
\end{proof}

\begin{lemma}
For all $E$, $E_1$, $E_2$, and $\Gamma$:
\begin{enumerate}
    \item $\qqpi{B(\Gamma)}{E} = 0 \Leftrightarrow{} not (\qqpi{B(\Gamma)}{E} = 1)$
    \item $\qqpi{B(\Gamma)}{E_1} = 1 \mathrel{and} \qqpi{B(\Gamma)}{E_2} = 1 \Leftrightarrow{}  \qqpi{B(\Gamma)}{E_1} \cdot \qqpi{B(\Gamma)}{E_2} = 1$
    \item $\qqpi{B(\Gamma)}{E_1} = 1 \mathrel{or} \qqpi{B(\Gamma)}{E_2} = 1 \Leftrightarrow{}  \qqpi{B(\Gamma)}{E_1} + \qqpi{B(\Gamma)}{E_2} - \qqpi{B(\Gamma)}{E_1} \cdot \qqpi{B(\Gamma)}{E_2} = 1$
\end{enumerate}
\end{lemma}
\begin{proof}
  These follow by cases analyses based on Lemma~\ref{lem:binary}.
\end{proof}

\begin{lemma}\label{lem:models}
    For all $E$ and $\Gamma$, either $\Gamma \models E$ or $\Gamma \models{} \mathrel{not} E$.
\end{lemma}

\begin{proof}
    By structural induction on the satisfaction relation.
\end{proof}



\restate{Theorem~\ref{thm:soft2hard}} For all $E$ and $\Gamma$: $\qqpi{B(\Gamma)}{E} = 1 \Leftrightarrow \Gamma \models E$.
\begin{proof}
  We prove the property by \emph{structural induction}; that is,
  we prove that $\phi(N)$ holds for all $N$, where $\phi(N)$ is as follows.
        \begin{equation*}
            \phi(N) \triangleq
                \forall E, \forall \Gamma :
                size(E)=N \Rightarrow (\Gamma \models E \Leftrightarrow \qqpi{B(\Gamma)}{E} = 1).
        \end{equation*}
        
  We proceed by course-of-values induction on $N$.
  Consider any $E, \Gamma$ and $N = size(E)$. We proceed by a case analysis at $E$.
  \begin{description}
      \item[$Base\;Case$]
                  For $N=1$, the base case is $E = (x_v \mathrel{is} l_\tau)$.
                  For any $\Gamma$ we are to show
                  \begin{equation*}
                  \Gamma \models x_v \mathrel{is} l_\tau \Leftrightarrow{}\qqpi{B(\Gamma)}{x_v \mathrel{is} l_\tau} = 1.
                  \end{equation*}
                  
                  By definition, $\qqpi{B(\Gamma)}{x_v \mathrel{is} l_\tau} = p_{v,\tau}$
                  where $p_{v,\tau}$ is the probability that variable $x_v$ has type $l_\tau$ according to the matrix $B(\Gamma)$. 
                  By definition of $B(\Gamma)$ and because $B$ results to  a binary matrix,  
                  $\qqpi{B(\Gamma)}{x_v \mathrel{is} l_\tau} = 1$
                  means that the element $p_{v,\tau}$ is equal to $1$, that is 
                  $\Gamma \models x_v \mathrel{is} l_\tau$.
                  Also, $\Gamma \models x_v \mathrel{is} l_\tau$, implies that $\Gamma(x_v)=l_\tau$. By definition, that
                  means  $\qqpi{B(\Gamma)}{x_v \mathrel{is} l_\tau} = 1$.
      
     \item[$Case \; E ={} \mathrel{not} E'$.]
                  We are to show
                  $\Gamma \models{} \mathrel{not} E' \Leftrightarrow \qqpi{B(\Gamma)}{\mathrel{not} E'} = 1$.
                  We have that,
                  \begin{align*}
                        \qqpi{B(\Gamma)}{\mathrel{not} E'} = 1 & \Leftrightarrow{} &  \\
                        1- \qqpi{B(\Gamma)}{E'} = 1 & \Leftrightarrow{} & \text{(Definition)}  \\
                        \qqpi{B(\Gamma)}{E'} = 0 & \Leftrightarrow{} &  \\
                        \mathrel{not} (\qqpi{B(\Gamma)}{E'} = 1) & \Leftrightarrow{} & \text{(Lemma~\ref{lem:binary})} \\
                        \mathrel{not}  \Gamma \models{}  E'
                        & \Leftrightarrow{} & \text{(Induction Hypothesis)} \\
                         \Gamma \models{} \mathrel{not} E' & & \text{(Definition)}.
                  \end{align*}
          
          \item[$Case \; E = (E_1 \mathrel{and} E_2)$.] 
                  For $size(E_1)<N$ and $size(E_2) < N$, we are to show that
                  $\Gamma \models (E_1 \mathrel{and} E_2) \Leftrightarrow \qqpi{B(\Gamma)}{(E_1 \mathrel{and} E_2)} = 1$. Our induction hypothesis is that $\phi(M)$ holds for all $M<N$. We have that
                  \begin{align*}
                      \Gamma \models (E_1 \mathrel{and} E_2)                            & \Leftrightarrow &  \\
                      \Gamma \models E_1 \mathrel{and} \Gamma \models E_2               & \Leftrightarrow & \text{(Definition)}  \\
                      \qqpi{B(\Gamma)}{E_1} = 1 \mathrel{and} \qqpi{B(\Gamma)}{E_2} = 1 & \Leftrightarrow & \text{(Induction Hypothesis)}           \\
                      \qqpi{B(\Gamma)}{E_1} \cdot \qqpi{B(\Gamma)}{E_2} = 1             & \Leftrightarrow &  \text{(Lemma~\ref{lem:binary})}           \\                 
                      \qqpi{B(\Gamma)}{(E_1 \mathrel{and} E_2)} = 1                  &   &  \text{(Definition)}                                         \\
                  \end{align*}
                  which completes the proof for this case.

            \item[$Case \; E = (E_1 \mathrel{or} E_2)$.]
                  For $size(E_1) < N$ and $size(E_2) < N$, we are to show
                  $\Gamma \models (E_1 \mathrel{and} E_2) \Leftrightarrow \qqpi{B(\Gamma)}{(E_1 \mathrel{and} E_2)} = 1$. Our induction hypothesis is that $\phi(M)$ holds for all $M<N$. We have that
                  \begin{align*}
                      \Gamma \models (E_1 \mathrel{or} E_2)                                                                     & \Leftrightarrow & \\
                      \Gamma \models E_1 \mathrel{or} \Gamma \models E_2                                                        & \Leftrightarrow &\text{(Definition)} \\
                      \qqpi{B(\Gamma)}{E_1} = 1 \mathrel{or} \qqpi{B(\Gamma)}{E_2} = 1                                          & \Leftrightarrow & \text{(Induction Hypothesis)}                                  \\
                      (\qqpi{B(\Gamma)}{E_1} - 1) \cdot (1- \qqpi{B(\Gamma)}{E_2}) = 0                                          & \Leftrightarrow                                  \\
                      \qqpi{B(\Gamma)}{E_1} - \qqpi{B(\Gamma)}{E_1} \cdot \qqpi{B(\Gamma)}{E_2} + \qqpi{B(\Gamma)}{E_2} - 1 = 0 & \Leftrightarrow                         &\text{(Case Analysis \& Lemma~\ref{lem:binary})}         \\
                      \qqpi{B(\Gamma)}{E_1 \mathrel{or} E_2} = 1                                                               &. 
                  \end{align*}
                  which completes the proof for this case.
    \end{description}
\end{proof}

\section{Appendix: Neural Model }\label{app:neural} 
In this appendix we present the implementation details of the deep neural  used in \cref{ssec:natprodts}.
\begin{lstlisting}[numbers=none,caption={Our Character Level \textit{LSTM} model.},captionpos=b]
  LSTMClassifier(
  (embedding): Embedding(90, 128)
  (lstm): LSTM(128, 64)
  (hidden2out): Linear(in_features=64, 
                out_features=100, bias=True)
  (softmax): LogSoftmax()
  (optimization fun): ADAM)
\end{lstlisting}

\end{document}